\documentclass[journal,10pt]{IEEEtran}

\ifCLASSINFOpdf
\usepackage[pdftex]{graphicx}

\else

\fi

\usepackage[cmex10]{amsmath}

\usepackage{bbm}
\usepackage{float}
\usepackage{algorithmic}

\usepackage{array}

\usepackage{color}
\usepackage{amsfonts}
\usepackage{changebar}
\usepackage{graphicx}
\usepackage{epsfig}
\usepackage{amsthm}

\usepackage{epstopdf}
\usepackage{wrapfig}
\usepackage{subfig}
\usepackage{caption}
\usepackage{algorithm}
\usepackage{algorithmic}
\usepackage{tabularx}
\usepackage[table]{xcolor}

\newtheorem{theorem}{Theorem}

\newtheorem{remark}{Remark}
\newtheorem{proposition}{Proposition}

\DeclareMathOperator*{\argmax}{\arg\!\max}

\begin{document}

\title{Prospect Pricing in Cognitive Radio Networks}

\author{\IEEEauthorblockN{Yingxiang~YANG\IEEEauthorrefmark{1}~\and~
Leonard T. Park\IEEEauthorrefmark{1}~\and~Narayan B. Mandayam\IEEEauthorrefmark{1}~\and~Ivan Seskar\IEEEauthorrefmark{1}~\and\newline Arnold Glass\IEEEauthorrefmark{2}~\and~
Neha Sinha\IEEEauthorrefmark{2}}

\IEEEauthorblockA{\IEEEauthorrefmark{1}Department of ECE, Rutgers University\\Yangyx891121@gmail.com, narayan@winlab.rutgers.edu}
    
\IEEEauthorblockA{\IEEEauthorrefmark{2}Department of Psychology, Rutgers University\\aglass@rci.rutgers.edu, nehasinha132@gmail.com}
}
\IEEEspecialpapernotice{(Invited Paper)}

\maketitle

\begin{abstract}
Advances in cognitive radio networks have primarily focused on the design of spectrally agile radios and novel spectrum sharing techniques that are founded on Expected Utility Theory (EUT). In this paper, we consider the development of novel spectrum sharing algorithms in such networks taking into account human psychological behavior of the end-users, which often deviates from EUT. Specifically, we consider the impact of end-user decision making on pricing and management of radio resources in a cognitive radio enabled network when there is uncertainty in the Quality of Service (QoS) guarantees offered by the Service Provider (SP). Using Prospect Theory (a Nobel-Prize-winning behavioral economic theory that captures human decision making and its deviation from EUT), we design data pricing and channel allocation algorithms for use in cognitive radio networks by formulating a game theoretic analysis of the interplay between the price offerings, bandwidth allocation by the SP and the service choices made by end-users. We show that, when the end-users under-weight the service guarantee, they tend to reject the offer which results in under-utilization of radio resources and revenue loss. We propose prospect pricing, a pricing mechanism that can make the system robust to decision making and improve radio resource management. We present analytical results as well as preliminary human subject studies with video QoS.
\end{abstract}
\begin{IEEEkeywords}
Game Theory, Prospect Theory, Probability Weighting, Prospect Pricing
\end{IEEEkeywords}

\section{Introduction}

\IEEEPARstart{C}{ognitive} Radio Networks (CRNs) \cite{haykin2005cognitive} and advanced spectrum sharing techniques have been studied extensively over the past decade \cite{biglieri2012principles}. In general, game theory plays a major role in studying the economical effects that CRNs could bring to the Service Providers (SPs), as well as the optimal radio resource management for the SPs when designing spectrum sharing rules and algorithms. Some examples of the applications of game theory in CRNs include auction-based spectrum sharing \cite{huang2006auction}, data pricing \cite{niyato2008competitive}, power control and allocation \cite{chen2008cognitive}\cite{bloem2007stackelberg}, Quality of Service (QoS) management \cite{wang2010admission}\cite{southwell2014quality}, and security \cite{clancy2008security}.

Among the aspects mentioned above, QoS guarantee and data pricing are particularly interesting to us. One reason is that QoS is often hard to guarantee in CRNs, mainly due to spectrum uncertainty under a CRN setting \cite{6133750}\cite{pan2014spectrum}\cite{DBLP:journals/corr/NadendlaBV14}. For example, the uncertainty in available spectrum due to interfering users (or even primary users) can result in situations  where the service cannot be guaranteed for the time period required by the user. When the SP opportunistically acquires available spectrum by performing spectrum sensing herself \cite{pan2014spectrum}, a mis-detection will cause the user to experience large noise and low service guarantee when accessing the channel. In fact, the result of a Federal Communications Commission (FCC) survey, which aims to provide the users information on the service qualities of offerings by different SPs when making their decisions to purchase, has shown that the advertised transmission rate (which affects the QoS) was not 100\% guaranteed even in the broadband internet \cite{FCC1}\cite{FCC2}. The above issues naturally lead to the problem of data pricing, since the guarantee of the service quality contributes to the users' decision making process. Furthermore, research has shown that a user's subjective perceptions of the service quality often deviates from the actual service quality \cite{kim2010study}\cite{hossfeld2011quantification}\cite{schatz2011poor}. This indicates that pricing should not be entirely based on the QoS without taking the users' subjective perceptions of the service into consideration.

An even more important reason that motivates this work is that we believe that end-user behavior plays an important role under a CRN setting, and many algorithms designed for CRN can be potentially impacted by those behaviors. Examples include situations where a secondary user needs to decide whether or not to access spectrum based on the uncertainty in the spectrum sensing performed. Alternately, when a primary user chooses to lease her unoccupied spectrum to secondary users by algorithms based on non-cooperative games or auction mechanisms, the secondary users have to decide on whether or not to lease the spectrum, and how much to pay for it given the uncertainty surrounding the QoS. The above scenarios demand an understanding and accurate modeling of an end-user's decision making process, so that the primary users, when leasing their spectrum resources, can more accurately evaluate their expected outcomes.

This leads to the basic structure of our work. We investigate a secondary system, where an SP acquires bandwidth from primary users, and sells broadband internet service to end-users. In particular, we assume that the service cannot be fully guaranteed, and we model the uncertainty involved in the guarantee of the service with the probability that the service quality actually meets the advertised service quality and assume that this piece of information is available to the end-users when they make decisions, an idea inspired by \cite{FCC1}\cite{FCC2}. Next, we model the impact of end-user's decision making process using Prospect Theory (PT) \cite{kahneman1979prospect}, a Nobel-Prize-winning theory that is particularly successful in modeling and explaining how people's decisions under risks and uncertainty deviate from the framework of Expected Utility Theory (EUT) \cite{von2007theory}. We study the impact of the end-users' decision making process on the profit and radio resource management of the SP, when they have a skewed view on the service guarantee. To combat this impact, we propose prospect pricing, which focuses on possible strategies for the SP including bandwidth reallocation, rate control, bandwidth expansion/reduction, and admission control, all of which can be achieved under a CRN setting, and study their capabilities of recovering the revenue for the SP. Our results relate the SP's bandwidth resources with her ability to dynamically manage her radio resources so as to obtain the same amount of revenue as originally anticipated without considering the end-users' skewed perceptions. We show that under some conditions, the impact of the end-users' perception is large enough so that the SP simply cannot obtain the amount of revenue originally anticipated.

The rest of the paper is organized as follows. In Section II, we introduce the related work on data pricing, the background on PT, as well as the works that applies PT to wireless communications scenarios. In Section III, we model the interactions between the end-users and the SP as a Stackelberg game, while the conditions under which the existence of a pure strategy NE can be guaranteed are discussed in Section IV. In section V, we discuss the impact of the Probability Weighting Effect (PWE) on the end-user's decision making process, the revenue of the SP. Section VI discusses the prospect of recovering the revenue of the EUT game via prospect pricing. Numerical results are shown in Section VII while in section VIII, we discuss psychophysics experiments with human subjects of video QoS over wireless channels so as to model the parameter used to characterize the probability weighting effect.

\section{Related Work}
\subsection{Prospect Theory: a brief introduction}

The rationality assumption in game theory \cite{osborne1994course}, which states that a player's decision making process is often assumed to be completely following the axioms and theorems established in Expected Utility Theorey \cite{von2007theory}, has long been questioned by behavioral science \cite{einhorn1981behavioral}. Although EUT explains most of the people's decision making successfully, paradoxes have been observed in real life that contradict the predictions of EUT. Alternative theories explaining human's decision making processes were raised in the 1970s, with the most successful one being Prospect Theory \cite{kahneman1979prospect}, whose main differences with EUT are
\begin{enumerate}
\item Probability Weighting Effect (PWE): the weight of the payoff of each possible outcome is different from the probability of the occurrence of that outcome.
\item Framing Effect (FE): the payoff of each outcome is framed into either gain or loss relative to a reference point.
\end{enumerate}
These two features can be illustrated with a variation of the famous Allais's Paradox \cite{allais1953comportement}, which is also used in \cite{kahneman1979prospect}.

In the experiment, two problems were sequentially presented to a group of 100 participants. Each problem contains 2 alternatives. For the first problem, the participants were asked to choose between
\begin{itemize}
\item A: \$2500 with probability 0.33; \$2400 with probability 0.66; \$0 with probability 0.01;
\item B: \$2400 with certainty,
\end{itemize}
while in the second problem between
\begin{itemize}
\item C: \$2500 with probability 0.33; \$0 with probability 0.67;
\item D: \$2400 with probability 0.34; \$0 with probability 0.66.
\end{itemize}

According to EUT, the expected utility of each alternative can be calculated by taking the expectation of payoff amount for different outcomes, which, for an alternative with $M$ outcomes $o_1$ to $o_M$ and their corresponding occurring probabilities $o_1$ to $o_M$, can be computed with
\begin{align}
U_{EUT}=\sum_{i=1}^{M}v_{EUT}(o_i)p_i.
\end{align}
It can be easily verified that
\begin{align}
U_{EUT}(A)&=2500\times 0.33+2400\times 0.66+0\times 0.01\notag\\&=2409>2400=U_{EUT}(B),
\end{align}
while
\begin{align}
U_{EUT}(C)&=2500\times 0.33=825<850\notag\\&=2500\times 0.34=U_{EUT}(D).
\end{align}
Thus, if the participants make their decisions following the prediction of EUT, i.e., choosing the alternative that maximizes the expected utility, then the participants should prefer A to B in problem 1 and D to C in problem 2. However, the result shows that the majority (82\%) of the participants chose B in problem 1 and the majority (83\%) of the participants chose C in problem 2. 

The results violate the predictions of EUT, but under the framework of PT, they can be well explained. Under PT, people are assumed to choose the alternative that maximizes the prospect, which can be computed with
\begin{align}
U_{PT}=\sum_{i=1}^{M} v_{PT}(o_i)w(p_i).
\end{align}
The definition of prospect is very similar to the definition of the expected utility, except that $p_i$ is weighted by an inverse-S-shaped Probability Weighting Function (PWF) $w(\cdot)$, which characterizes the PWE analytically. In addition, $v_{EUT}(o_i)$ is replaced by $v_{PT}(o_i)$, which depicts the FE. Figure \ref{fig:pwf} illustrates the idea of PWE by Prelec's PWF. The PWF captures the feature that people often over-weight low probabilities and under-weight moderate and high probabilities. The value function captures the effect of loss aversion on people, i.e., the same amount of loss usually looms larger than the same amount of gain to a person.

\begin{figure}[!t]
  \centering
  \includegraphics[width=0.95\linewidth]{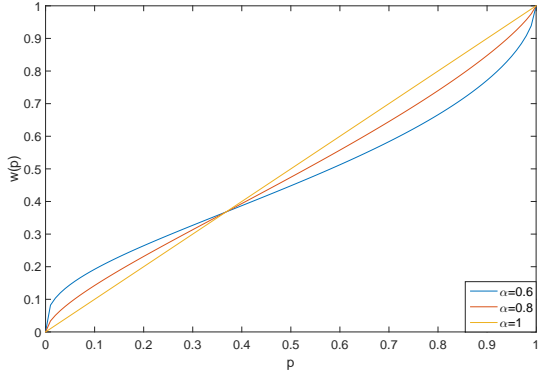}
  \caption{Prelec's probability weighting function with three different values $\alpha$ indicating different levels of skewness.}
  \label{fig:pwf}
\end{figure}

The result of the experiment can be explained immediately with the above setup. In problem 1, since alternative B provides a guaranteed payoff, that payoff becomes the reference point when framing the payoff of each outcome under the other alternative. Thus, \$2500 becomes a gain of \$100, while \$0 becomes a loss of \$2400. It can then be readily seen that if the probability 0.01 is over-weighted as depicted in Figure \ref{fig:pwf}, then most people would have indeed preferred B to A. The same argument applies to problem 2. In our work, we adopt Prelec's PWF, which is first proposed in \cite{Prelec}, and parametrized by $\alpha\in (0,1]$:
\begin{align}
w(p)=\exp\left\{-(-\ln p)^{\alpha})\right\}.
\end{align}
However, most of the conclusions in this work can be easily generalized to other inverse-S-shaped PWFs. We also report on psychophysics studies with human subjects to experimentally determine the value of $\alpha$ in section \ref{sec:psy}. Note that we do not consider the role of  FE in this paper and that is a topic for future study.

\subsection{Data pricing}

Pricing of wireless data has been widely studied for over a decade. Most of the work focuses on proposing mechanisms that offer control over the network's traffic while maximizing the revenue of the service provider. A comprehensive survey of the most typical strategies adopted by the SPs over the past years, offering either wired or wireless services, can be found in \cite{sen2012pricing}. Traditionally, the SPs use flat-rate pricing strategies as well as usage-based pricing strategies, which offer limited ability on managing network traffic. More complicated pricing strategies are adopted later on, for example Paris Metro pricing \cite{odlyzko1999paris}, time-of-day pricing \cite{ha2012tube}\cite{joe2011time}, and congestion level based pricing \cite{paschalidis2000congestion}. Those strategies are harder to implement, but offer better performances in managing the congestion level of the network, as well as higher service guarantee as they make some users back off when purchasing the service by making them aware of the actual cost of accessing the network when the congestion level is high by setting a higher price.

However, even with advanced pricing strategies, the uncertainty involved in the guarantee of the service cannot be avoided. In particular, in wireless communication, the uplink and downlink rates cannot be guaranteed due to noise and interference, which cannot be accurately predicted at the time the service is purchased. Thus, the end-users often have to make difficult choices between several alternatives of accessing the network, in which the service quality she gets is stochastic. 

Recently, there has been a category of work that study this particular type of decision making problem of the end-users with Prospect Theory, spanning a number of areas including communication networks \cite{PTRAGc,PTRAG,Spectrum,Tussle,Altruism,Adhoc,Antijamming}, and smart energy management \cite{micro}\cite{integrating}. The subject of pricing is addressed in \cite{PTRAGc}\cite{PTRAG}, and in our previous work \cite{CISS}\cite{Allerton}. In \cite{CISS}, we studied the same problem of this work under a more specific setting, i.e., assuming there exists only one end-user. We studied the conditions under which an NE exists, and found the NE that gives the SP maximum revenue. We then studied the case when the end-user follows the decision making process of PT, and showed that the SP cannot avoid revenue loss if she wants to retain the same NE or the same revenue under the PT game. 
In \cite{Allerton}, we generalized the framework to the multiuser setting.

\section{A Stackelberg Game Model}

\begin{table*}
\centering
\begin{tabular}{|p{3cm}|p{6cm}|p{3.5cm}|}
\hline Parameter & Meaning & Location \\\hline
$\{b,r_{EUT}(b),\vec{BW}_{EUT}\}$ & service offering & page 3  \\\hline
$BW_{max,EUT}$ & total bandwidth & page 4 \\\hline $B_i$ & actual rate (a random variable) of $i$-th user & page 4 \\\hline $\bar{F}_{B_i}(b;BW_{i,EUT})$ & service guarantee for the $i$-th user & equation (6) \\\hline
$h_i(b)$ & $i$-th user's benefit function & page 4 \\\hline $p_i$ & $i$-th user's probability of accepting the offer & page 4 \\\hline
$b_{1,EUT}^*$ & rate offered under Nash Equilibrium & equation (11) \\\hline $\bar{p}$ & average acceptance probability & Remark 1 \\\hline $\lambda_i$, $\lambda_{i,AD}$, $\lambda_{i,RC}$ & dummy variables & Theorem 2, Propositions 1,3 \\\hline $S_{EUT}$, $S_{PT}$ & The set of users to which the service is offered & equation (15) \\\hline $BW_{max,PT}$ & total bandwidth constrained under PT and EUT games & equation (15) \\\hline $r_{PT}(b)$ & price at rate $b$ under PT game & page 6 \\\hline $b^*_{1,PT}$ & rate under Nash Equilibrium in PT game & equation (15) \\\hline $\vec{BW}_{PT}$ & bandwidth allocation under PT game & equation (15)\\\hline Simulation parameters & & TABLE II, page 9\\\hline
\end{tabular}
\caption{Parameters involved and the locations of their definitions and introductions}
\label{tab:par}
\end{table*}

\begin{figure}[!t]
\centering
\includegraphics[width=\linewidth]{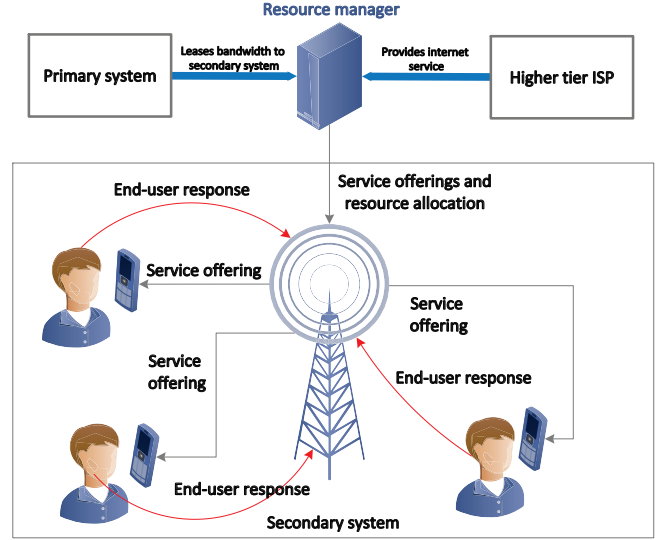}
\caption{System model}
\label{fig:system}
\end{figure}

We study the scenario under a CRN shown in Figure \ref{fig:system}, where the secondary system has a service manager, or SP, who actively manages and allocates available radio resources and sells service to $N$ end-users. The scenario where $N=1$ is a special case and is studied in \cite{CISS}. The bandwidth is assumed to be obtained from primary users by means of trading, an assumption that's frequently considered to maximize the bandwidth utilization \cite{jayaweera2009dynamic}\cite{mwangoka2011broker}. Meanwhile, the data is assumed to be obtained from the service offerings of higher tier ISPs. The interaction between the SP and the end-users is modeled into a Stackelberg game. The SP is aware of the number of users within its service range, and moves first by investing necessary resources, and making offers to the users. The end-users then decide whether or not to accept the service. The decisions are assumed to be made at the same time.

We define an offer made by the SP under the EUT game as a triple $\{b,r_{EUT}(b),\vec{BW}_{EUT}\}$, which corresponds to the rate $b$, the price of the service at that rate determined by a predefined pricing function $r_{EUT}(b)$, and a specific allocation of the SP's bandwidth denoted by $\vec{BW}_{EUT}=\{BW_{1,EUT},...,BW_{N,EUT}\}$, which satisfies the total bandwidth constraint $\vert\vec{BW}_{EUT}\vert=BW_{max,EUT}$. On the user side, we assume that the actual service rate for the $i$-th user is a random variable $B_i$, and the service guarantee at rate $b$ is a function that only depends on the channel between the user and the SP, the rate offered and the amount of bandwidth allocated to that user, and has the following form
\begin{align}
\bar{F}_{{B_i}}(b;BW_{i,EUT}):=\mathbb{P}({B_i}>b\vert BW_{i,EUT}).
\end{align}
For a fixed rate $b$, increasing $BW_{i,EUT}$ raises the service guarantee.

Denote the $i$-th user's benefit upon receiving guaranteed service at rate $b$ with $h_i(b)$. Since the SP offers constant rate, we can see that if the user accepts the offer, she pays a price $r_{EUT}(b)$, and with probability $\bar{F}_{B_i}(b;BW_{i,EUT})$ she receives successful service, and with probability $F_{B_i}(b;BW_{i,EUT}):=1-\bar{F}_{B_i}(b;BW_{i,EUT})$ the channel cannot successfully deliver the service at rate $b$ and the user experiences an outage. Hence, denoting the acceptance probability of the $i$-th user as $p_i$, the expected utility of the $i$-th user can then be represented as
\begin{align}
\label{eq:Uuser}
&U_{user,i}(p_i,b,BW_{i,EUT})=p_i[-r_{EUT}(b)+\notag\\+&h_i(b)\bar{F}_{B_i}(b;BW_{i,EUT})+h_i(0)F_{B_i}(b;BW_{i,EUT})]\notag\\+&(1-p_i)h_i(0).
\end{align}
As a natural assumption, we assume that $h_i(0)=0$ for all users. Thus, $U_{user,i}(p_i,b,BW_{i,EUT})=p_i[-r_{EUT}(b)+h_i(b)\bar{F}_{B_i}(b;BW_{i,EUT})]$. Note that the above model of the user's utility function is a special case of a more general scenario where the SP constantly adapts her transmission rate according to the channel's capacity. Under this general setting, $U_{user,i}(p_i,b,BW_{i,EUT})=p_i[\mathbb{E}[h_{i}(B_i)]-r_{EUT}(b)]+(1-p_i)h_{i}(0)$.
This general form of the user's expected utility reduces to (\ref{eq:Uuser}) when the user's benefit function is a step function, i.e., $h_{i}(B_i)=h_i(b)>0$ for $B_i>b$ and $h_i(0)$ for $B_i<b$ with $b$ being the advertised rate by the SP, which is equivalent to assuming that the user is insensitive to the actual service rate higher than advertised, and is extremely sensitive when the service delivered is below rate $b$. More general cases involving more complicated form of $h_i(B_i)$ can be studied, but involves tedious work on analyzing the properties of $\mathbbm{E}[h_i(B_i)]$ under the EUT game and PT game\footnote{The computation of expectation of a continuous random variable under the probability weighting effect can be dealt with the help from \cite{tversky1992advances}\cite{rieger2008prospect}.}.

As for the SP, a cost $c_i(b,BW_{i,EUT})$ is incurred upon her when she makes an offer at rate $b$ to the $i$-th user. Specifically, we assume an affine cost function for each individual user
\begin{align}
c_i(b,BW_{i,EUT})=c_1b+c_3BW_{i,EUT},
\end{align}
since the SP invests in resources based on the number of users in its service range. $c_1$ and $c_3$ are the cost for unit data rate and bandwidth. The fixed cost for the SP is ignored. Hence, the expected utility of the SP can be expressed as
\begin{align}
&U_{SP}(\vec{p},b,\vec{BW}_{EUT})=\sum_{i=1}^{N}[p_i[r_{EUT}(b)-c_i(b,BW_{i,EUT})]\notag\\+&(1-p_i)(-c_i(b,BW_{i,EUT}))].
\end{align}

We place a few more natural assumptions on our model. Firstly, $r_{EUT}(b)$ and $h_i(b)$ are assumed to be monotonically increasing and concave. The service guarantee for each user is assumed to converge to 0 as the offered rate tends to $\infty$ under fixed bandwidth. Meanwhile, the service guarantee for a user is a monotonically increasing function with respect to the bandwidth allocated to that user.

Lastly, we summarize the parameters we use in Table \ref{tab:par}.

\section{Existence of multiple Nash Equilibria of the EUT game}

With the above settings, the conditions for the existence of an NE can be characterized. Consider two cases, with one involving only a single user, and the other involving multiple users. For simplicity, we dub the first case as a Single-User-Single-Provider (SUSP) game, and the second case as a Multiple-User-Single-Provider (MUSP) game.

\begin{theorem}[The existence of multiple Nash Equilibria (NE)] 
\label{theorem1}
Assuming that $\forall i$,
\begin{align}
r_{EUT}(b^*_{1,EUT})>c_i(b^*_{1,EUT},BW_{max,EUT}),
\end{align}
where
\begin{align}
b^*_{1,EUT}=\argmax_{b} (r_{EUT}(b)-c_i(b,BW_{max,EUT})),
\end{align}
then there exists a pure strategy NE for the MUSP game\footnote{We use the phrase ``pure strategy NE" to refer the NE where the users accept the service with probability 1. The case where the users decline the service offer is excluded from our context.} if and only if there exists a pure strategy NE for at least one of the SUSP game consisting of one of the $N$ users and the SP, under which the SP allocates the entire bandwidth she has to that user.
\end{theorem}

The detailed proof can be found in \cite{Allerton}.

\begin{remark}
It is worth pointing out that we do not consider mixed strategy NE in the MUSP game. This is because, assuming that under an NE the acceptance probability of the users is represented by $\vec{p}$, the offered rate is $b$, and the allocation of the bandwidth is represented by $\vec{BW}_{EUT}$, we have $U_{SP} (\vec{p},b,\vec{BW}_{EUT})=\sum_{i\in S_{EUT}} p_i(r_{EUT}(b)-c_i(b,BW_{i,EUT}))=\bar{p}\left\vert S_{EUT}\right\vert r_{EUT}(b)-\left\vert S_{EUT}\right\vert c_1b-c_3BW_{max,EUT}$, where $\bar{p}$ is the average acceptance probability of all the users within set $S_{EUT}$. In order to reach a mixed strategy NE, the SP must find a rate $b$ and a corresponding bandwidth allocation $\vec{BW}_{EUT}$ such that all the users are indifferent between accepting and denying the offer. However, the expression also shows that the users' acceptance probabilities represented by $\vec{p}$ only affect the SP's decisions through their average $\bar{p}$. Hence, for any combinations of offered rate and bandwidth allocation that induce a mixed strategy NE, the acceptance probabilities of the users can be arbitrary as long as the average acceptance probability remains fixed and the SP cannot obtain a higher revenue through offering the service to a subset of $S_{EUT}$. Hence, the SP does not have control over the individual user's acceptance probability under a mixed strategy NE.
\end{remark}

We next specifies a procedure with which the SP finds the strategy that leads to the revenue-maximizing NE. This strategy includes a service offering, and the corresponding bandwidth allocation. We define
the minimum amount of bandwidth that can be allocated to user $i$ at rate $b$ to be
\begin{align*}
BW_i(b)=\bar{F}^{-1}_{B_i}\left(\frac{r_{EUT}(b)}{h_i(b)},b\right),
\end{align*}
which is equivalent to the amount of bandwidth that satisfies
\begin{align*}
\bar{F}_{B_i}(b;BW_i(b))=\frac{r_{EUT}(b)}{h_i(b)},
\end{align*}
and we assume that the SP knows this piece of information for all users. The procedure is specified in algorithm \ref{alg:NE}.

In algorithm \ref{alg:NE}, the SP first categorizes all pure strategy NE according to the number of users that accepts the offer. For an NE where $n$ users accept the offer, the SP goes on to find the rate such that when $n$ users accept the offer, the revenue is maximized while the minimum bandwidth needed to support the service to the selected users satisfies the total bandwidth constraint. However, if there exists an $S$ with a larger size, then the strategy will not lead to an NE. Hence we set the revenue to 0 so that it will not be selected. The revenues under different choices of $n$ are then compared, and the revenue maximizing NE is located for a specific $n^*$. Lastly, the SP selects a possible choice of $S$ with size $n^*$, and allocate slightly more than the minimum bandwidth needed for each user so that they accept the offer with probability 1. Finally, note that under the revenue-maximizing $n^*$, the choice of $S$ is unique, which consists of the $n^*$ users with the lowest values of $BW_i(b^*_{1,EUT})$.

\begin{algorithm}[H]
\caption{Locating the revenue-maximizing NE}
\label{alg:NE}
    \begin{algorithmic}[1]
        \STATE {\it Input:} $S_{EUT}$, $r_{EUT}(b)$, $BW_{max,EUT}$, and for all $i$'s $h_i(b)$ and $BW_i(b)$.
        \STATE {\it Output:} The revenue maximizing strategy $(b^*_{1,EUT},r_{EUT}(b^*_{1,EUT}),\vec{BW}_{EUT})$.
        \FOR {$n=|S_{EUT}|$ to 1}
            \STATE {$b^*_{1,EUT}[n]\leftarrow \argmax_{b} nr_{EUT}(b)-nc_1b$, s.t., $\min_{S\subseteq S_{EUT},|S|=n}\sum_{j\in S}BW_j(b)< BW_{max,EUT}$.}
            \STATE {$U_{SP,EUT}[n]\leftarrow nr_{EUT}(b^*_{1,EUT}[n])-nc_1b^*_{1,EUT}[n]-c_3BW_{max,EUT}$}
            \IF {$\exists S\subseteq S_{EUT}, |S|>n$, such that $\sum_{j\in S}BW_j(b^*_{1,EUT})<BW_{max,EUT}$}
                \STATE {$U_{SP,EUT}[n]\leftarrow 0$}
            \ENDIF
        \ENDFOR
        \STATE {$n^*\leftarrow\argmax_{n}U_{SP,EUT}[n]$}
        \STATE {$b^*_{1,EUT}\leftarrow b^*_{1,EUT}[n^*]$}
        \STATE {$S^*\leftarrow \arg_{S,|S|=n}\sum_{j\in S}BW_j(b^*_{1,EUT})< BW_{max,EUT}$}
        \STATE {$BW_{i}\leftarrow BW_{i}(b^*_{1,EUT})+\epsilon$ if $i\in S$, and $BW_{i}\leftarrow 0$ otherwise}
    \end{algorithmic}
\end{algorithm}

\section{The impact of Prospect Theory on end-user decisions}

In the remainder of this paper, we consider the impact of Prospect Theory on end-users' decisions of whether or not to accept a service offer, its impact on the radio resources and the revenue of the SP. In particular, we focus on the effect of end-user's weighting of the service guarantee, i.e., the PWE aspect of PT. We shall see that, when the end-users under-weight the service guarantee, they tend to reject the offer, which leads to an under-utilization of the SP's radio resources and a loss in revenue.

For the MUSP game, we study the condition under which the system is robust to the PWE in the sense of retaining all the users without having to change the service offer. The result is summarized as follows.

\begin{theorem}
\label{theorem3}
If all the users under-weight the service guarantee, and the same offer inducing the pure strategy NE under the EUT game is offered to the same set of users, then the NE is preserved under PWE if and only if $\forall i\in S_{EUT}$,
\begin{align}
\label{sufficient}
BW_{i,EUT}>\bar{F}^{-1}_{B_i}\left(\lambda_i,b^*_{1,EUT}\right),
\end{align}
where
\begin{align}
\lambda_i = w^{-1}\left(\bar{F}_{B_i}\left(b^*_{1,EUT};BW_i(b^*_{1,EUT})\right)\right).
\end{align}
\end{theorem}
\begin{proof}
For the $i$-th user, the necessary and sufficient condition for him to accept an offer at rate $b^*_{1,EUT}$ and price $r_{EUT}(b^*_{1,EUT}$ under the impact of PWE is $h_i(b^*_{1,EUT})w(\bar{F}_{B_i}(b^*_{1,EUT};BW_{i,EUT}))>r_{EUT}(b^*_{1,EUT})$, i.e., $w(\bar{F}_{B_i}(b^*_{1,EUT};BW_{i,EUT}))>\bar{F}_{B_i}(b^*_{1,EUT};BW_i(b^*_{1,EUT}))$. Since the probability weighting function is monotonically increasing, and thus have an inverse, we have $BW_{i,EUT}>\bar{F}^{-1}_{B_i}\left(w^{-1}\left(\bar{F}_{B_i}\left(b^*_{1,EUT};BW_i(b^*_{1,EUT})\right)\right),b^*_{1,EUT}\right)$.
\end{proof}

The above result indicates that, in order to retain all the users without changing the service offer, the total amount of bandwidth of the SP must be ``sufficient":
\begin{align}
\label{sufficienttotal}
BW_{max,EUT}&=\sum_{i\in S_{EUT}} BW_{i,EUT}\notag\\&>\sum_{i\in S_{EUT}}\bar{F}^{-1}_{B_i}\left(\lambda_i,b^*_{1,EUT}\right).
\end{align}
When $\alpha=1$, $w(p)=p$, and the PT game reduces to EUT game. As $\alpha$ decreases, $w^{-1}(p)$ increases for every fixed $p$ that satisfies $w(p)<p$, and hence the right hand side of the above inequality increases, indicating that when PWE is introduced and the users under-weight the service guarantee, the SP must invest in more bandwidth than the amount required under the EUT game in order to retain all the users with the same offer.

\section{Prospect Pricing}

In this section, we introduce the idea of prospect pricing to make the system robust against the PWE experienced by the users. For the MUSP game, the SP needs to perform prospect pricing by setting a new price $r_{PT}(b)$ at the offered rate $b$ when the bandwidth of the system does not satisfy the condition specified in equation (\ref{sufficienttotal}). The goal of prospect pricing consists the following two aspects.

\begin{itemize}
\item Retain the Radio Resource Management (RRM) constraints when PWE is introduced. The RRM constraints for the MUSP game are defined as follows
\begin{align}
\left\{\begin{array}{l}
S_{EUT}=S_{PT},\\
b^*_{1,EUT}=b^*_{1,PT},\\
BW_{max,EUT}=BW_{max,PT},\\
\vec{BW}_{EUT}=\vec{BW}_{PT}\end{array}\right..
\end{align}

The constraints restrict the SP to offer service of the same rate to the same set of users when PWE is introduced. They also restrict the SP to allocate the same amount of bandwidth to each user within the set.

\item Retain the revenue of the EUT game when the end-users under-weight the service guarantee.
\end{itemize}

We first show that, in the MUSP game the SP cannot retain her revenue and the RRM constraints simultaneously, provided that all the users under-weight the service guarantee, i.e., the SP cannot strictly retain all RRM constraints without suffering a revenue loss. We then show that by partially relaxing the RRM constraints, it is possible for the SP to retain her revenue under the EUT game.

\begin{theorem}
\label{theorem4}
When (\ref{sufficienttotal}) is not satisfied, and when all the users under-weight the service guarantee, Prospect Pricing can be used to retain strict RRM constraints, at the cost of the SP losing revenue of at least
\begin{align*}
L_{RRM}&=\max_{i\in S_{EUT}}\{r_{EUT}(b^*_{1,EUT})-\\&-h_i(b^*_{1,EUT})w(\bar{F}_{B_i}(b^*_{1,EUT};BW_{i,EUT}))\}.
\end{align*}
\end{theorem}

The detailed proof can be found in \cite{Allerton}.

Next, we discuss four other strategies that can be used by the SP along with prospect pricing. These include
\begin{itemize}
\item {\bf\emph{Bandwidth reallocation}}: the SP reallocates the available unoccupied bandwidth among the users. In a CRN, the SP is capable of performing this since she needs to reallocate her bandwidth whenever a channel allocated to a secondary user is occupied by a primary user.
\item {\bf\emph{Admission control}}: the SP offers the service to a set of users $S_{PT}$ which is a subset of $S_{EUT}$.
\item {\bf\emph{Bandwidth expansion/reduction}}: the SP invests in a different amount of bandwidth $BW_{max,PT}$. This can be achieved when the spectrum of the SP is leased from primary users, which has been a commonly adopted assumption \cite{pan2014spectrum}.
\item {\bf\emph{Rate control}}: the SP offers a different rate $b^*_{1,PT}$ to the users, similar to rate adaptation often used in CRN algorithms.
\end{itemize}
Note that, except for bandwidth expansion/reduction, the other strategies requires maintaining the total bandwidth constraint, i.e., $BW_{max,EUT}=BW_{max,PT}$. The allocation of the bandwidth among the end-users, however, can be arbitrary. When (\ref{sufficienttotal}) is not satisfied, we want to find out whether the above four strategies, when applied together with prospect pricing, could help the SP retain the revenue she would get if the users follows decision making process of under the EUT framework. The results are described below.

\subsection{Bandwidth reallocation}

In bandwidth reallocation, the SP has the freedom to change the amount of bandwidth allocated to each user, subject to the total bandwidth constraints. The rate offered must also be the same. 

\begin{theorem}
With bandwidth reallocation, the revenue loss can be reduced, but not fully recovered.
\end{theorem}
\begin{proof}
In order to retain strict RRM constraints, all the users must accept the same offer containing the same rate and bandwidth, i.e., $\forall i\in S_{EUT}$, $r_{PT}(b^*_{1,EUT})<h_i(b^*_{1,EUT})w(\bar{F}_{B_i}(b^*_{1,EUT};BW_{i,EUT}))$. Hence, we have $r_{PT}(b^*_{1,EUT})<$ $\min_{i\in S_{EUT}} \{h_i(b^*_{1,EUT})w(\bar{F}_{B_i}(b^*_{1,EUT};$ $BW_{i,EUT}))\}$. However, $\min_{i\in S_{EUT}} \{h_i(b^*_{1,EUT})w(\bar{F}_{B_i}$ $(b^*_{1,EUT};BW_{i,EUT}))\}<r_{EUT}(b^*_{1,EUT})$. Hence, in order to retain strict RRM constraints, the SP must take a revenue loss of at least $L_{RRM}=r_{EUT}(b^*_{1,EUT})-\min_{i\in S_{EUT}}\{h_i(b^*_{1,EUT})w(\bar{F}_{B_i}(b^*_{1,EUT};BW_{i,EUT}))\}=\max_{i\in S_{EUT}}\{r_{EUT}(b^*_{1,EUT})-h_i(b^*_{1,EUT})w(\bar{F}_{B_i}(b^*_{1,EUT};$ $BW_{i,EUT}))\}$. Allowing reallocation of the bandwidth will reduce the revenue loss, since the revenue loss allowing bandwidth allocation $L_{BA}$ is the minimum revenue loss over all possible bandwidth allocation, and the bandwidth allocation under strict RRM constraints is only one instance.

We next show that allowing reallocation of the bandwidth cannot help the SP to fully recover the revenue by contradiction. Since the service is offered to the same set of users and the offered rate remains the same, the price must be the same in order to retain the revenue, i.e., $r_{EUT}(b^*_{1,EUT})=r_{PT}(b^*_{1,EUT})$. Assume that there exists a bandwidth allocation such that $\forall i\in S_{EUT}$, $
r_{PT}(b^*_{1,EUT})<h_i(b^*_{1,EUT})w(\bar{F}_{B_i}(b^*_{1,EUT};BW_{i,PT}))$. Then we must have $\forall i\in S_{EUT}$,
\begin{align}
BW_{i,PT}&>\bar{F}^{-1}_{B_i}\left(w^{-1}\left(\frac{r_{PT}(b^*_{1,EUT})}{h_i(b^*_{1,EUT})}\right);b^*_{1,EUT}\right)\notag\\&=\bar{F}^{-1}_{B_i}\left(\lambda_i;b^*_{1,EUT}\right).
\end{align}
Hence, the summation over the set $S_{EUT}$ yields the condition specified in (\ref{sufficienttotal}), contradicting the assumption that the bandwidth is insufficient in the first place.
\end{proof}

Note that the SP can acquire the optimal bandwidth allocation under the PT game by minimizing $L_{RRM}$ with respect to the amount of bandwidth allocated to each user. Hence the reduction of $L_{RRM}$ is 0 if the original bandwidth allocation under the EUT is the same as this optimal bandwidth allocation scheme.

Next, we explore other ways of relaxing the RRM constraints the SP can resort to in order to recover her revenue. We discuss a set of necessary and sufficient conditions under which the revenue can be recovered.

\subsection{Admission control}

The SP is allowed to violate the RRM constraints by selecting $S_{PT}\subset S_{EUT}$. Upon excluding one user, the SP is able to reallocate the bandwidth to other users to increase service performance.

\begin{proposition}
The necessary and sufficient condition for the SP to recover her revenue is to have sufficient bandwidth under the EUT game. Mathematically,
\begin{align}
BW_{max,EUT}&=\sum_{i\in S_{PT}} BW_{i,PT}\notag\\&>\min_{S_{PT}\subset S_{EUT}}\sum_{i\in S_{PT}}\bar{F}^{-1}_{B_i}\left(w^{-1}\left(\lambda_{i,AD}\right);b^*_{1,EUT}\right),
\end{align}
with
\begin{align}
\lambda_{i,AD}=\frac{\frac{|S_{EUT}|}{|S_{PT}|}r_{EUT}(b^*_{1,EUT})-\left(\frac{|S_{EUT}|}{|S_{PT}|}-1\right)c_1b^*_{1,EUT}}{h_i(b^*_{1,EUT})}.
\end{align}
\end{proposition}
\begin{proof}
We start by showing necessity. In order to retain revenue, we must have $|S_{EUT}|(r_{EUT}(b^*_{1,EUT})-c_1b^*_{1,EUT})-c_3BW_{max,EUT}=|S_{PT}|(r_{PT}(b^*_{1,PT})-c_1b^*_{1,PT})-c_3BW_{max,PT}$, where $b^*_{1,EUT}=b^*_{1,PT}$, and $BW_{max,EUT}=BW_{max,PT}$. Hence, $\exists r_{PT}(b^*_{1,PT})$ and $S_{PT}$ s.t., $|S_{EUT}|(r_{EUT}(b^*_{1,EUT})-c_1b^*_{1,EUT})=|S_{PT}|(r_{PT}(b^*_{1,PT})-c_1b^*_{1,PT})$, i.e., the form of the pricing function under the PT game at rate $b^*_{1,PT}$ is
\begin{align}
\label{eq:adctrl}
r_{PT}(b^*_{1,PT})&=\frac{|S_{EUT}|}{|S_{PT}|}r_{EUT}(b^*_{1,EUT})-\notag\\&-\left(\frac{|S_{EUT}|}{|S_{PT}|}-1\right)c_1b^*_{1,EUT}.
\end{align}
In order for the users to accept the offer, we must have $\forall i \in S_{PT}$,
\begin{align}
\label{accept}
r_{PT}(b^*_{1,PT})<h_i(b^*_{1,PT})w(\bar{F}_{B_i}(b^*_{1,PT};BW_{i,PT})),
\end{align}
i.e., $\forall i\in S_{PT}$, the amount of bandwidth under the PT game is lower bounded by
\begin{align}
BW_{i,PT}&>\bar{F}^{-1}_{B_i}\left(w^{-1}\left(\frac{r_{PT}(b^*_{1,PT})}{h_i(b^*_{1,PT})}\right);b^*_{1,PT}\right)\notag\\&=\bar{F}^{-1}_{B_i}\left(w^{-1}\left(\lambda_{i,AD}\right);b^*_{1,EUT}\right).
\end{align}
Since the total bandwidth is constrained to $BW_{max,EUT}$, we have
\begin{align}
BW_{max,EUT}&=\sum_{i\in S_{PT}} BW_{i,PT}\notag\\&>\sum_{i\in S_{PT}}\bar{F}^{-1}_{B_i}\left(w^{-1}\left(\lambda_{i,AD}\right);b^*_{1,EUT}\right),
\end{align}
which can be further lower bounded by taking the minimum over all $S_{PT}\subset S_{EUT}$.

We next prove sufficiency. Firstly, if the above condition is satisfied, then we must have a set $S_{PT}\subset S_{EUT}$ and a corresponding price under the NE $r_{PT}(b^*_{1,PT})$ such that $\forall i\in S_{PT}$, equation (\ref{accept}) holds. This is obvious, as we can simply choose the price as given by equation (\ref{eq:adctrl}) and allocate $BW_{i,PT}$ slightly higher than the minimum amount required.

Hence, the only thing left to show is that
\begin{align}
|S_{PT}|(r_{PT}(b^*_{1,PT})-c_1b^*_{1,PT})>c_3\sum_{i\in S_{PT}}BW_{i,PT}.
\end{align}

This is also true since the left hand side minus the right hand side is just the revenue of the SP under the EUT game. By assumption, this revenue must be positive.

\end{proof}

\subsection{Bandwidth Expansion/Reduction}

As suggested by the name, the SP is allowed to violate the RRM constraints such that $BW_{max,PT}\neq BW_{max,EUT}$. By doing this, the SP is also allowed to reallocate her bandwidth among the users. A set of necessary and sufficient conditions is given as follows.

\begin{proposition}
The necessary and sufficient condition for the SP to recover her revenue under the EUT game is that she has sufficient bandwidth under the EUT game. Mathematically,
\begin{align}
&BW_{max,EUT}>\frac{1}{c_3}\{|S_{EUT}|r_{EUT}(b^*_{1,EUT})-\sup_{BW_{PT}}\sup_{\vec{BW}_{PT}}\notag\\&\min_{i\in S_{EUT}} [|S_{EUT}|h_i(b^*_{1,EUT})w\left(\bar{F}_{B_i}(b^*_{1,EUT};BW_{i,PT})\right)-\notag\\&-c_3BW_{PT}]\},
\label{eq:bwexpthreshold}
\end{align}
where the outer supremum is to find the optimal total amount of bandwidth $BW_{PT}$ under the PT game. The inner supremum is to find the optimal bandwidth allocation subject to the constraint that the total amount of bandwidth under the PT game is $BW_{PT}$.
\end{proposition}
\begin{proof}
We start by showing necessity. In order to retain revenue, we must have $|S_{EUT}|(r_{EUT}(b^*_{1,EUT})-c_1b^*_{1,EUT})-c_3BW_{max,EUT}=|S_{PT}|(r_{PT}(b^*_{1,PT})-c_1b^*_{1,PT})-c_3BW_{max,PT}$, where $S_{EUT}=S_{PT}$, and $b^*_{1,EUT}=b^*_{1,PT}$. Hence, the form of the pricing function under the PT game at rate $b^*_{1,PT}$ is $r_{PT}(b^*_{1,EUT})=r_{EUT}(b^*_{1,EUT})+\frac{c_3}{|S_{EUT}|}(BW_{max,PT}-BW_{max,EUT})$. Once again, equation (\ref{accept}) must hold $\forall i\in S_{EUT}$ in order for the users to accept the offer, hence $\forall i\in S_{EUT}$,
\begin{align}
\label{eq:necbwexp}
r_{EUT}(b^*_{1,EUT})+\frac{c_3}{|S_{EUT}|}(BW_{max,PT}-BW_{max,EUT})\notag\\<h_i(b^*_{1,EUT})w(\bar{F}_{B_i}(b^*_{1,EUT};BW_{i,PT})).
\end{align}
This implies that the left hand side of the above equation must be smaller than the minimum of the right hand side with respect to $i$. It also implies that there exists a way of allocating bandwidth $BW_{max,PT}$ under the PT game such that $BW_{max,EUT}>\frac{1}{c_3}\{|S_{EUT}|r_{EUT}(b^*_{1,EUT})-\min_{i\in S_{EUT}}[|S_{EUT}|h_i(b^*_{1,EUT})w(\bar{F}_{B_i}(b^*_{1,EUT};BW_{i,PT}))-c_3BW_{max,PT}]\}$, which further implies (\ref{eq:bwexpthreshold}).

We next show sufficiency. Similar to the case of admission control, the form of $r_{PT}$ has been specified. Also, there exists a way of allocating the bandwidth among the users such that equation (\ref{eq:necbwexp}) holds $\forall i\in S_{PT}$. This can be achieved by simply choosing $w(\bar{F}_{B_i}(b^*_{1,EUT};BW_{i,PT}))$ to be the minimizing solution for the right hand side of (\ref{eq:bwexpthreshold}).

\end{proof}

\begin{remark}
We can obtain a result parallel to the above proposition, which bounds the maximum amount of allowed skewness of the PWE of the users given the bandwidth of the SP. From equation (\ref{eq:necbwexp}), we can equivalently have $h_i(b^*_{1,EUT})w(\bar{F}_{B_i}(b^*_{1,EUT};BW_{i,PT}))-\frac{c_3}{|S_{EUT}|}BW_{max,PT}>r_{EUT}(b^*_{1,EUT})-\frac{c_3}{|S_{EUT}|}BW_{max,EUT}$. Hence, for each $BW_{max,PT}$ we requires the existence of an allocation of the bandwidth under the PT game, such that we have $r_{EUT}(b^*_{1,EUT})-c_3BW_{max,EUT}<\sup_{\vec{BW}}\min_{i} [|S_{EUT}|h_i(b^*_{1,EUT})w(\bar{F}_{B_i}(b^*_{1,EUT};BW_{i,PT}))-c_3BW_{max,PT}]$, where $\vec{BW}$ subjects to the total bandwidth constraints under the PT game. For any $BW_{max,PT}$, the optimal way of maximizing the left hand side of the above inequality by allocating the bandwidth is to make the weighted guarantee the same for all the users. Hence, upon assuming this guarantee is $x$, the above relationship becomes $|S_{EUT}|h_i(b^*_{1,EUT})x-c_3\sum_{i\in S_{EUT}}\bar{F}_{B_i}^{-1}(w^{-1}(x);b^*_{1,EUT})>r_{EUT}(b^*_{1,EUT})-c_3BW_{max,EUT}$. Hence, in order for a valid $BW_{max,PT}$ to exist, we must have $c_3<\sup_{x}\frac{|S_{EUT}|h_i(b^*_{1,EUT})x-(r_{EUT}(b^*_{1,EUT})-c_3BW_{max,EUT})}{\sum_{i\in S_{EUT}}\bar{F}_{B_i}^{-1}(w^{-1}(x);b^*_{1,EUT})}$, where $x$ is constrained to the region where probability is under-weighted, i.e., $(e^{-1},1]$.

Note that, although the constraint is on $c_3$ (which appears on both sides of the condition), it actually reveals a constraint on $w$. If no PWE is involved, i.e., $w(p)=p$, then this constraint is always satisfied. This is because we can select $x$ such that the denominator on the right hand side is the total bandwidth of the SP under the EUT game. In this case, $x$ is lower bounded by the minimum service guarantee of all the users under the NE of the EUT game. However, as the PWE sets in, the denominator of the right hand side increases monotonically, indicating that in order for the SP to recover the revenue under the EUT game, the user's cannot have a too skewed perception of the probability. In case of Prelec's PWF, it means that for every value of $BW_{max,EUT}$, there is a minimum $\alpha$ below which the SP is unable to recover her revenue. This is also a necessary and sufficient condition for the SP to recover her revenue under the EUT game.
\end{remark}

\subsection{Rate control}

Lastly, we consider the option of rate control, which allows the SP to optimize over the rate she offers to the users, the bandwidth allocation, but constraining the total bandwidth to be the same as in the EUT game. Here, a necessary and sufficient condition is specified as follows.

\begin{proposition}

A necessary and sufficient condition for the SP to recover her revenue under the EUT game is that the SP has sufficient bandwidth under the EUT game. Mathematically,
\begin{align}
BW_{max,EUT}&=\sum_{i\in S_{EUT}}BW_{i,PT}\notag\\&>\inf_{b_{1,PT}}\sum_{i\in S_{EUT}}\bar{F}^{-1}_{B_i}\left(w^{-1}\left(\lambda_{i,RC}\right);b_{1,PT}\right),
\end{align}
where
\begin{align}
\lambda_{i,RC}=\frac{r_{EUT}(b^*_{1,EUT})+c_1(b_{1,PT}-b^*_{1,EUT})}{h_i(b_{1,PT})}.
\end{align}

\end{proposition}

\begin{proof}

Starting from the same equation, in order for the SP to recover the revenue, there must exist a price under the NE of the PT game such that $|S_{EUT}|(r_{EUT}(b^*_{1,EUT})-c_1b^*_{1,EUT})-c_3BW_{max,EUT}=|S_{PT}|(r_{PT}(b^*_{1,PT})-c_1b^*_{1,PT})-c_3BW_{max,PT}$, where now we have $S_{EUT}=S_{PT}$ and $BW_{max,EUT}=BW_{max,PT}$. Thus, $r_{EUT}(b^*_{1,EUT})-c_1b^*_{1,EUT}=r_{PT}(b^*_{1,PT})-c_1b^*_{1,PT}$, and the form of the pricing function under the PT game at rate $b^*_{1,PT}$ is $r_{PT}(b^*_{1,PT})=r_{EUT}(b^*_{1,EUT})+c_1(b^*_{1,PT}-b^*_{1,EUT})$. Once again, the condition for the users to accept the offer (\ref{accept}) must apply. Hence, $\forall i\in S_{EUT}$, $r_{EUT}(b^*_{1,EUT})+c_1(b^*_{1,PT}-b^*_{1,EUT})<h_i(b^*_{1,PT})w(\bar{F}_{B_i}(b^*_{1,PT};BW_{i,PT}))$. Thus,
\begin{align}
BW_{i,PT}&>\bar{F}^{-1}_{B_i}\left(w^{-1}\left(\lambda_{i,RC}\right);b^*_{1,PT}\right),
\end{align}
indicating that
\begin{align}
BW_{max,EUT}&=\sum_{i\in S_{EUT}}BW_{i,PT}\notag\\&>\inf_{b_{1,PT}}\sum_{i\in S_{EUT}}\bar{F}^{-1}_{B_i}\left(w^{-1}\left(\lambda_{i,RC}\right);b_{1,PT}\right).
\end{align}

Here $b_{1,PT}>0$ and satisfies the constraint such that $h_i(b_{1,PT})-c_1b_{1,PT}<r_{EUT}(b^*_{1,EUT})-c_1b^*_{1,EUT}$, but can be further constrained if we desire to avoid letting the SP to recover the revenue by offering the users a substantially lower rate.

The sufficiency follows immediately if we choose $b_{1,PT}$ to be the minimizing solution of the above equation.

\end{proof}

\begin{remark}

The result shown above automatically implies that the probability weighting function of the users cannot be too skewed. In case of Prelec's PWF, as $\alpha$ decreases, the right hand side monotonically increases, showing that for every original design of the EUT game, there is a minimum $\alpha$ below which the SP cannot recover her revenue.
\end{remark}

\subsection{Summary}

The result shown by the previous subsections also indicates that the maximum amount of revenue that could be retained by combining each of the strategies of bandwidth expansion/reduction, rate control, admission control with prospect pricing is different. Since the relative performance of the three methods above can be evaluated by comparing the minimum amount of bandwidth needed for the SP to recover her anticipated revenue, the method that corresponds to the lowest threshold is most robust. With further assumptions regarding the forms of the service guarantee and the user's benefit, we compare the performance of those three strategies numerically in the next section.

\section{Numerical Results}
\subsection{Experiment setup}
In this section, we demonstrate some of the conclusions drawn above. We consider a scenario where $N=10$ users are spread uniformly within a single cell with a radius of 800 meters. There are no interference between different users, and we assume that the SP offers the service to all the users. Each user experiences a combination of path loss, shadowing, and Rayleigh fading. The guarantee of the service for each user in this setup is one minus the outage probability of the fading channel between the user and the base station and the rate offered is the encoding rate at the transmitter. The path loss and shadowing are calculated using a simplified model \cite{Goldsmith}
\begin{align}
P_{r_i}=P_{t}+K-\gamma \log_{10}\frac{d}{d_0}+\varphi_{i,dB},
\end{align}
where $P_{t}$ and $P_{r_i}$ are the transmitted signal power and the received signal power at the $i$-th user in decibels, $K$ is a constant taking the value $-20\log_{10} (4\pi d_0/\lambda)$. $\gamma$ is the path loss exponent, $d$ is the distance between the user and the base station antenna, and $d_0$ is the reference distance for the antenna far-field. In addition, $\varphi_{i,dB}$ is a Gaussian random variable that captures the effect of shadow fading. Finally, the guarantee function for each user can be expressed as
\begin{align}
\bar{F}_{B_i}(b)=\exp\left\{-\frac{2^{\frac{b}{BW_{i,EUT}}}-1}{P_{r_i}/(N_0BW_{i,EUT})}\right\}.
\end{align}
where $N_0$ is the power spectral density (PSD) of the noise, $BW_{i,EUT}$ represents the bandwidth allocated to the $i$-th user, and $b$ represents the encoding rate of the SP.

A list of the values for the parameters can be found in the following table.
\begin{table}
\label{simpar}
\begin{center}
\caption{Parameters used for simulation}
\begin{tabular}{|c|c|c|}

\hline Parameter & Meaning & Value \\\hline
$P_t$ & Transmission power & $10$ W \\\hline
 $K$ & Antenna dependent constant & $-64.5$ dB \\\hline
 $N_0$ & PSD of thermal noise & $-174$ dBm \\\hline
 $d_0$ & Reference distance for the antenna far-field & $20$ m \\\hline
 $\gamma$ & Path loss exponent & $4$ \\\hline
 $\sigma$ & Standard deviation for $\varphi_{i,dB}$ & 4 \\\hline
 $r$ & Cell radius & $800$ m \\\hline
\end{tabular}
\end{center}
\end{table}

\subsection{Bandwidth reallocation}
In Figure ~\ref{RevenueLoss}, the minimum revenue loss with and without enforcing strict RRM constraints are shown. The horizontal axis represents different values of $\alpha$, the parameter that captures the level of probability weighting of the users, while the vertical axis represents the revenue loss normalized by the revenue the SP makes under the EUT game. The total amount of bandwidth is 10 percent more than the minimum bandwidth needed for all the users to accept the offer with probability 1, and is allocated in a way such that each user receives 10 percent more bandwidth than the minimum amount needed for her to accept the offer under the EUT game. The blue curve shows the minimum revenue loss when strict RRM constraints are enforced, while the red curve corresponds to the case where the SP is allowed to violate the RRM constraints by reallocating the bandwidth among the users. As can be seen from the graph, when $\alpha=1$, the users weight the service guarantee accurately, and no revenue is lost. As $\alpha$ decreases, the system with blue curve starts to lose revenue first, and the revenue loss is always higher than that corresponding to the system that allows bandwidth reallocation, which corresponds to our result in Theorem \ref{theorem4}. It can also be seen from the plot that the system with strict RRM constraints do not start losing revenue until $\alpha$ is smaller than 0.94. This is because of the extra 10 percent of bandwidth, which holds (\ref{sufficient}) for all the users when $\alpha\geq 0.94$. Finally, we point out that when $\alpha<0.93$, the difference between the revenue losses of the two systems are roughly 1 percent of the total revenue of the SP. Since $N=10$, this converts to roughly 10 percent of revenue the SP makes from a single user. The price at rate $b^*_{1,EUT}$ under the EUT and PT game is shown in Figure ~\ref{rpt}. We can see that the price reduction is smaller when the SP is allowed to re-allocate her bandwidth.

\begin{figure}[!t]
\centering
\includegraphics[width=\linewidth]{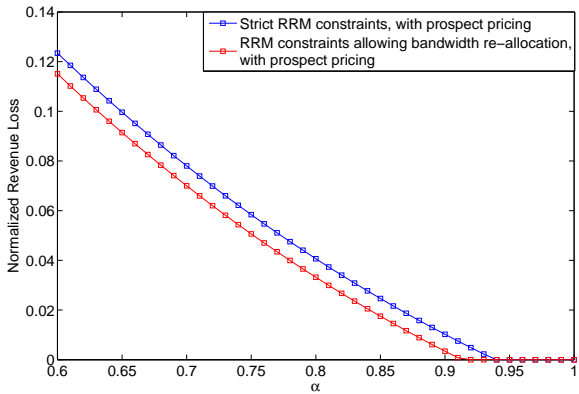}
\caption{Revenue loss of the SP normalized by the revenue under EUT game. $N=10$, $h_i(b)=10^{-2}\times (b\times 10^{-3})^{0.65}$, $r_{EUT}(b)=2\times10^{-3}\times(b\times10^{-3})^{0.82}$, $c_1=\frac{1}{3}\times 10^{-6}$, $c_3=10^{-8}$, $c_i(b;BW_{i,EUT})=c_1b+c_3BW_{i,EUT}$ $b^*_{1,EUT}\approx 7Mbps$, $BW_{max,EUT}\approx 14MHz$.}
\label{RevenueLoss}
\end{figure}

\begin{figure}[!t]
\centering
\includegraphics[width=\linewidth]{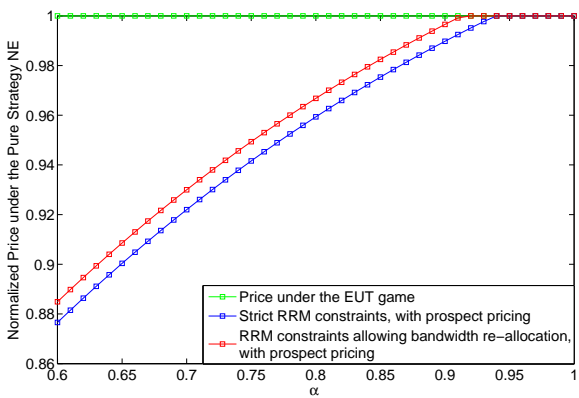}
\caption{Price of the SP at rate $b^*_{1,EUT}$ normalized by the price under the EUT game. $N=10$, $h_i(b)=10^{-2}\times (b\times 10^{-3})^{0.65}$, $r_{EUT}(b)=2\times10^{-3}\times(b\times10^{-3})^{0.82}$, $c_1=\frac{1}{3}\times 10^{-6}$, $c_3=10^{-8}$, $c_i(b;BW_{i,EUT})=c_1b+c_3BW_{i,EUT}$, $b^*_{1,EUT}\approx 7Mbps$, $BW_{max,EUT}\approx 14MHz$.}
\label{rpt}
\end{figure}

\subsection{Bandwidth expansion/reduction }

To illustrate the effect of the bandwidth expansion/reduction, we show minimum amount of bandwidth under the EUT game required for the SP to recover her revenue and normalize it by $BW_{max,EUT}$, and we show the corresponding maximum revenue under the PT game, normalized by the revenue of the SP under the EUT game. We also show $BW_{max,EUT}$ normalized by itself (which is equal to the horizontal line $f(x)=1$). It can be immediately seen from Figure \ref{fig:bandwidthexpansion1} that, when $\alpha$ is higher than 0.89, the maximum revenue under the PT game goes above 1 after normalization, which implies that the SP is able to recover her revenue under the PT game completely. The same threshold is also exactly the same crossing of the curves showing the minimum bandwidth requirement under the EUT game and the horizontal line showing the normalized system bandwidth under the EUT game. This illustrates our proposition, since on the right hand side of the crossing, the actual bandwidth of the SP under the EUT game is above the threshold, which implies that she is able to recover the revenue completely.
\begin{figure}
\centering
\includegraphics[width=\linewidth]{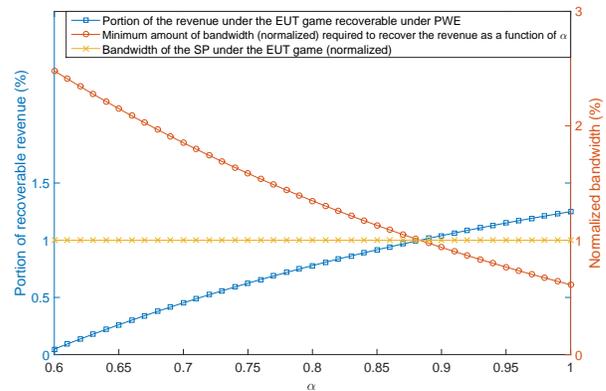}
\caption{With prospect pricing, the minimum amount of bandwidth needed to recover the revenue in full and the maximum amount of revenue attainable under the PT game with the total bandwidth constraint $BW_{max,EUT}$. }
\label{fig:bandwidthexpansion1}
\end{figure}


\subsection{Rate control}

It can be expected that the performance of the rate control would display a similar pattern to the results of bandwidth expansion/reduction. We thus show the results as a comparison to the other methods in last subsection.

\subsection{Admission control}

The result of admission control is shown in Figures \ref{fig:Adctrlprice} and \ref{fig:Adctrlloss}, where we have considered a 50-user scenario, and have plotted the pricing function of the SP under the NE versus $\alpha$ for different levels of admission control. Each time the SP applies the admission control to the current user set, she drops the user that consumes most bandwidth. It can be seen from Figure \ref{fig:Adctrlloss} that, when no admission control is applied, the SP suffers revenue loss. However, upon excluding one user, she is able to redistribute the bandwidth among the remaining users, raise the service guarantee for them. The SP is also able to mitigate the impact of the PWE and recover her revenue for alpha above a certain threshold for each different value of $|S_{PT}|$.
\begin{figure}[!t]
\centering
\includegraphics[width=\linewidth]{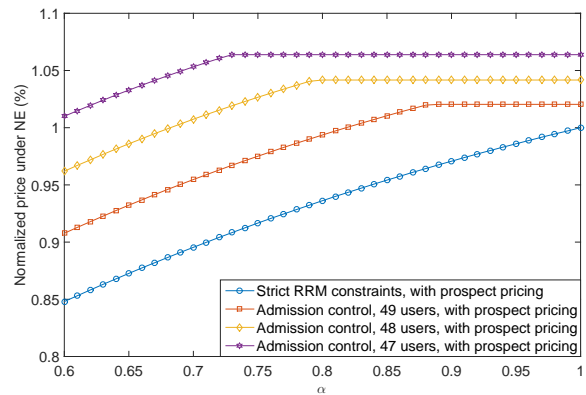}
\caption{Admission control applied to 50 users distributed in the cell, excluding up to 3 users.}
\label{fig:Adctrlprice}
\end{figure}
\begin{figure}[!t]
\centering
\includegraphics[width=\linewidth]{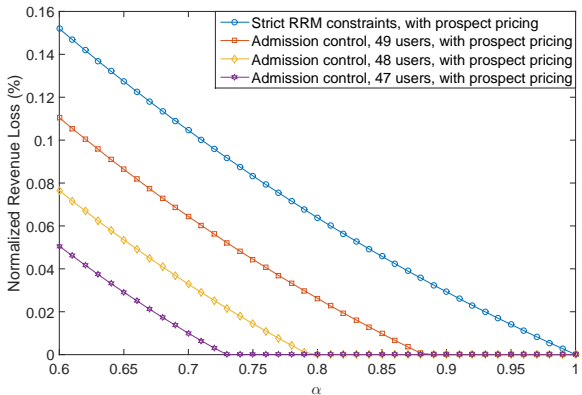}
\caption{Admission control applied to 50 users distributed in the cell, excluding up to 3 users.}
\label{fig:Adctrlloss}
\end{figure}

The characterization of the minimum amount of bandwidth is shown in next subsection.

\subsection{Performance comparison}

In Figure \ref{fig:perfcomp1}, we show the minimum amount of bandwidth needed for the three methods to help the SP retain revenue with prospect pricing, and show the minimum bandwidth needed in order for the SP to retain partial RRM constraints without prospect pricing. It can be immediately seen from the graph that the bandwidth expansion/reduction without prospect pricing requires the largest amount of bandwidth for low values of $\alpha$, indicating that it's the least robust against the probability weighting effect, and it cannot help the SP to completely recover her revenue under the EUT game as can be seen in Figure \ref{fig:perfcomp2}. Secondly, when $\alpha$ is below 0.96, there is no solution for admission control. This shows that admission control is not effective against low $\alpha$ when the number of users is low, since the spare bandwidth recycled from the denied user cannot efficiently raise the perceived service guarantee of the remaining users. However, admission control is able to recover the revenue when $\alpha$ is close to 1. Finally, in this particular case, the rate control is the most efficient method in recovering the revenue. Part of the reason is that the cost for data rate is higher than the cost for bandwidth, giving more freedom to the method of rate control.

\begin{figure}[!t]
\centering
\includegraphics[width=\linewidth]{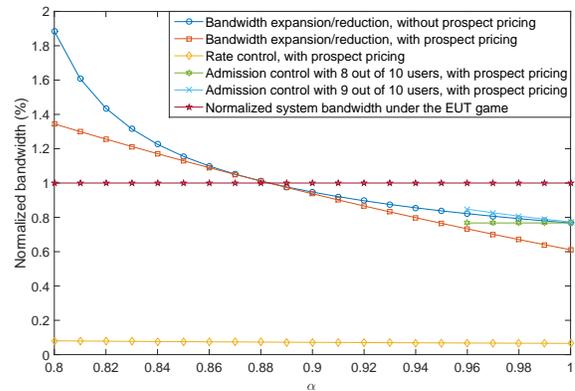}
\caption{Minimum bandwidth required for the SP to retain partial RRM constraints under the PT game without prospect pricing (bandwidth expansion/reduction without prospect pricing) and other methods to retain revenue}
\label{fig:perfcomp1}
\end{figure}

\begin{figure}[!t]
\centering
\includegraphics[width=\linewidth]{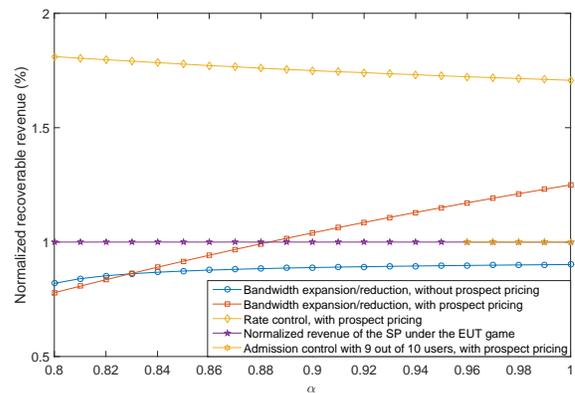}
\caption{Maximum normalized recoverable revenue for the SP through bandwidth expansion/reduction with/without prospect pricing, rate allocation, and admission control with 9 out of 10 users.}
\label{fig:perfcomp2}
\end{figure}

\section{Psychophysics experiments with video QoS}
\label{sec:psy}

In this section, we provide experimental data which supports the procedure of modeling the end-user's probability weighting effect with Prelec's PWF. Specifically, we conducted human subject studies as it relates to the perception of video service quality and then used these studies to estimate the parameter $\alpha$ that reflects people's weighting effect on the uncertainty in QoS. The experiment was conducted using a testbed shown in Figure \ref{fig:platform} with 23 psychology college students, where each subject is asked to assess the quality of a 1 hour video comprised of 30 2-minute segments, where each segment of the video is subject to different packet loss and delay parameters. The testbed comprises a single compute/communication device (the programmable ORBIT radio node \cite{raychaudhuri2005overview}) with two major software components (i) a network emulation module (NETEM), and (ii) a content caching module. The radio modem in the ORBIT node is used to implement a soft access point which end-users connect to. All the traffic coming to the access point is subject to traffic shaping policies as specified in the NETEM module, specifically to control wireless network performance in terms of packet loss and delay. To alleviate the artifacts of wide area internet connectivity on the experimental conditions, we logically created a local caching functionality in the platform. The end-user interface device is a laptop used to watch the video.

Using the testbed, for each pair of packet loss and delay chosen, we are able to objectively measure the corresponding decoded frames per second at the video player used to display the video. Our psychophysics experiments have revealed that the decoded video frames per second serves as the best objective proxy for the quality of the video among the parameters chosen, while the feelings about the number of stops and stutters occurred is the best proxy for the subjective ratings on the overall video quality. The human subjects are also asked to subjectively evaluate on a four level scale the quality of the video as they perceive it, with 4 being the highest rating and 1 being the lowest rating. 

Tables \ref{tab:subj} and \ref{tab:obj} show the subjective (on a scale of 1-4) and objective (decoded video frames per second) measurements along with their mean and standard deviation. As can be seen in the tables, the highest actual video quality corresponds to the unit in the upper left corner, where no packet loss and delay are present. The lowest video qualities being rated are the units just above the blackened out units. The blackened out area of the tables essentially refers to the situations where there quality of the wireless channel is so poor that there is no video displayed in the player. Even the raw data in terms of the subjective scores reveals that there is tendency of the human subjects to ``underweight" the best (even perfect) video quality and ``overweight" the worst case video quality. This effect can also be observed explicitly in Figure \ref{fig:cluster}, where we show the relationship between the subjective rating and the objective metric with 95\% confidence level. 

It also follows from the objective measurements in Table \ref{tab:obj}, that the x-axis in Figure 11 can be mapped directly as a proxy for the objective probability of service guarantee.  In order to map the relationship between objective and subjective probabilities to that of a Prelec-like PWE, as a first cut, we use a simple uniform mapping of the subjective measurements to the region $[0,1]$. The result is depicted in Figure \ref{fig:fitting}, where we obtain the probability of each frame being displayed successfully as $p$, and the probability of the participant believing that the video is uninterrupted as $w(p)$. The relationship between these two variables display an inverse S-shaped probability weighting effect. We fit a parametric function of the Prelec form to the above data set and the resulting parameter $\alpha$ that minimizes the mean-squared error (MSE) is found to be $\alpha \approx 0.585$.  

Note that there have been efforts to subjectively evaluate video QoS \cite{seshadrinathan2010study} that have used various technical measures such as peak signal to noise ratio (PSNR) but there have been none to evaluate the probability weighting effect (psychophysics function) such as undertaken here. The human subject studies presented here is the first such effort and will be further expanded to include larger data sets as well as more detailed mapping techniques to map objective and subjective measurements to the corresponding probabilities of service guarantees (uncertainty). Further, such psychophysics studies can also be conducted by the SP for learning each individual user's subjective perceptions to objective metrics and can be easily implemented via appropriate ``apps" on end-user devices such as smart phones.

\begin{figure}
\centering
\includegraphics[width=\linewidth]{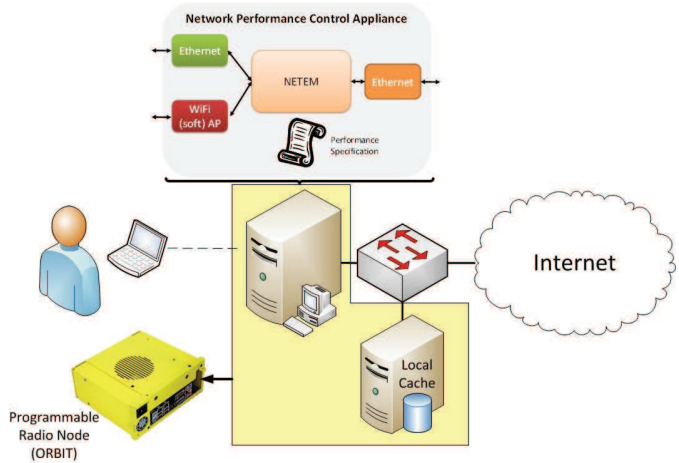}
\caption{Experimental platform illustration}
\label{fig:platform}
\end{figure}

\begin{table*}[]
\centering
\begin{tabular}{|c|c|c|c|c|c|c|c|c|c|c|c|}
\hline\multicolumn{12}{|c|}{Subjective Ratings}\\\hline
\multicolumn{2}{|c|}{Packet Loss (\%)} & \multicolumn{2}{|c|}{0} & \multicolumn{2}{|c|}{2} & \multicolumn{2}{|c|}{4} & \multicolumn{2}{|c|}{8} & \multicolumn{2}{|c|}{16} \\\hline
\multicolumn{2}{|c|}{Delay (ms)} & mean & dev & mean & dev & mean & dev & mean & dev & mean & dev \\\hline \multicolumn{2}{|c|}{0} & 3.67 & 0.64 & 3.71 & 0.56 & 3.60 & 0.58 & 3.16 & 1.02 & 1.29 & 0.45 \\\hline \multicolumn{2}{|c|}{20} & 3.80 & 0.40 & 2.89 & 1.15 & 2.66 & 1.10 & 1.54 & 0.77 & 1.16 & 0.36 \\\hline \multicolumn{2}{|c|}{40} & 3.82 & 0.48 & 2.00 & 0.85 & 1.54 & 0.86 & 1.27 & 0.54 & 1.10 & 0.30 \\\hline \multicolumn{2}{|c|}{60} & 3.58 & 0.76 & 1.71 & 0.87 & 1.34 & 0.57 & 1.25 & 0.54 & 1.21 & 0.52 \\\hline \multicolumn{2}{|c|}{80} & 3.69 & 0.71 & 1.40 & 0.77 & 1.21 & 0.41 & 1.10 & 0.44 & \cellcolor{black} & \cellcolor{black} \\\hline \multicolumn{2}{|c|}{160} & 3.53 & 0.63 & 1.17 & 0.48 & 1.24 & 0.64 & \cellcolor{black} & \cellcolor{black} & \cellcolor{black} & \cellcolor{black} \\\hline \multicolumn{2}{|c|}{320} & 2.65 & 0.92 & 1.10 & 0.29 & \cellcolor{black} & \cellcolor{black} & \cellcolor{black} & \cellcolor{black} & \cellcolor{black} & \cellcolor{black}\\\hline \multicolumn{2}{|c|}{640} & 1.87 & 0.99 & \cellcolor{black} & \cellcolor{black} & \cellcolor{black} & \cellcolor{black} & \cellcolor{black} & \cellcolor{black} & \cellcolor{black} & \cellcolor{black} \\\hline
\end{tabular}
\caption{Subjective Measurement}
\label{tab:subj}
\end{table*}

\begin{table*}[]
\centering
\begin{tabular}{|c|c|c|c|c|c|c|c|c|c|c|c|}
\hline\multicolumn{12}{|c|}{Decoded Video Frames per Second}\\\hline
\multicolumn{2}{|c|}{Packet Loss (\%)} & \multicolumn{2}{|c|}{0} & \multicolumn{2}{|c|}{2} & \multicolumn{2}{|c|}{4} & \multicolumn{2}{|c|}{8} & \multicolumn{2}{|c|}{16} \\\hline
\multicolumn{2}{|c|}{Delay (ms)} & mean & dev & mean & dev & mean & dev & mean & dev & mean & dev \\\hline \multicolumn{2}{|c|}{0} & 21.28 & 1.79 & 21.50 & 1.81 & 21.27 & 1.83 & 19.67 & 3.30 & 5.32 & 2.46 \\\hline \multicolumn{2}{|c|}{20} & 21.55 & 1.50 & 18.01 & 4.87 & 15.28 & 5.39 & 9.26 & 4.10 & 3.39 & 2.77 \\\hline \multicolumn{2}{|c|}{40} & 21.84 & 1.43 & 13.33 & 4.97 & 9.79 & 3.71 & 5.11 & 2.59 & 2.90 & 1.79 \\\hline \multicolumn{2}{|c|}{60} & 20.92 & 4.34 & 10.28 & 5.08 & 7.54 & 3.70 & 5.74 & 4.34 & 2.47 & 1.70 \\\hline \multicolumn{2}{|c|}{80} & 21.62 & 1.60 & 7.50 & 4.22 & 5.48 & 3.08 & 3.65 & 2.44 & \cellcolor{black} & \cellcolor{black} \\\hline \multicolumn{2}{|c|}{160} & 20.14 & 2.53 & 4.66 & 2.48 & 3.94 & 4.10 & \cellcolor{black} & \cellcolor{black} & \cellcolor{black} & \cellcolor{black} \\\hline \multicolumn{2}{|c|}{320} & 16.97 & 4.39 & 3.67 & 1.74 & \cellcolor{black} & \cellcolor{black} & \cellcolor{black} & \cellcolor{black} & \cellcolor{black} & \cellcolor{black}\\\hline \multicolumn{2}{|c|}{640} & 11.71 & 4.94 & \cellcolor{black} & \cellcolor{black} & \cellcolor{black} & \cellcolor{black} & \cellcolor{black} & \cellcolor{black} & \cellcolor{black} & \cellcolor{black} \\\hline
\end{tabular}
\caption{Objective Measurement}
\label{tab:obj}
\end{table*}

\begin{figure}[!t]
\centering
\includegraphics[width=\linewidth]{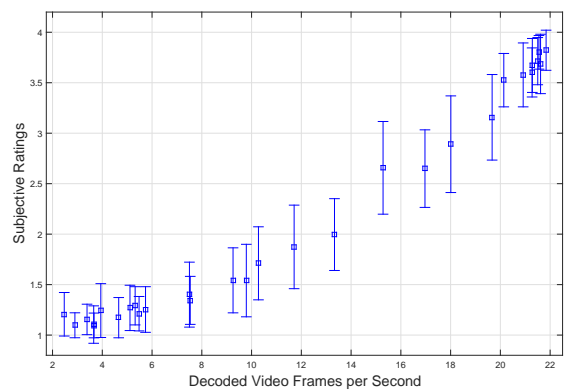}
\caption{Quality of service ratings shown as a function of decoded video frames per second with 95\% confidence level.}
\label{fig:cluster}
\end{figure}

\begin{figure}[!t]
\centering
\includegraphics[width=\linewidth]{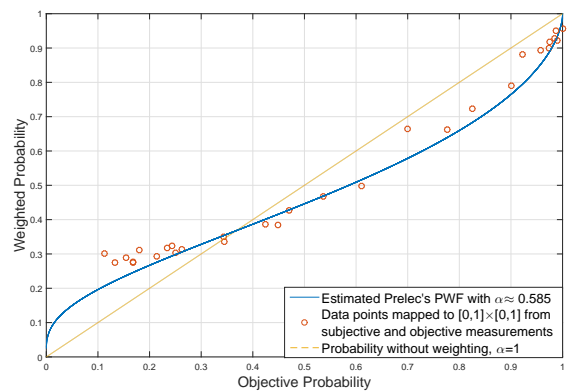}
\caption{The probability weighting effect can be well approximated with a Prelec's PWF with $\alpha\approx 0.585$.}
\label{fig:fitting}
\end{figure}

\section{Conclusion and Discussion}
In this work, we considered the impact of end-users' decisions in regard to service offers in a CRN when there is uncertainty in the QoS guarantee offered by the SP. We formulated a Stackelberg game to study the interplay between the price offerings, bandwidth allocation by the SP and the service choices made by end-users. We characterized the NE of the game, and showed that when the end-users under-weight the service guarantee, they tend to reject the service offers which results in under-utilization of radio resources and revenue loss for the SP. To combat this effect, we proposed prospect pricing, which combines the pricing strategy of the SP with the radio resource management strategy available under a CRN setting. In particular, we studied four distinct strategies, namely bandwidth reallocation, bandwidth expansion/reduction, rate control, admission control, and studied the capability of each individual strategy in helping to improve the utilization of radio resources and enable the SP to recover her revenue loss. Our results first show that the SP must have sufficient bandwidth in order to combat the under-weighting effect by the end-users without prospect pricing, and if the bandwidth is insufficient, then bandwidth reallocation alone cannot help the SP recover her revenue. As for the remaining three strategies, our results show that, for each individual strategy, there is a threshold dependent on the skewness of the end-users' PWF and the unit cost for data rate and bandwidth, such that in order for the SP to recover her revenue, her total bandwidth under the EUT game must be above this threshold. We also showed that having sufficient bandwidth that is above this threshold (dependent on the strategy taken) is also a necessary and sufficient condition for the SP to be able to recover her revenue. We also compared the performance of the bandwidth expansion/reduction, rate control and admission control strategies with numerical examples that illustrate the threshold effect described above. We also conducted psychophysics experiments with human subjects to assess perceived video QoS over wireless channels and  modeled the probability weighting effect. The focus of this paper has been on studying the PWE effects of end-user behavior and the role FE in influencing end-user behavior is a topic for future study.

\section{Acknowledgement}

This work is supported in part by a U.S. National Science Foundation (NSF) Grant (No. 1421961) under the NeTS program. The authors would like to thank Nilanjan Paul for assisting in the design of the psychophysics testbed, and Alysia Tsui for assisting in the human subjects studies.

\bibliography{biblio}{}
\bibliographystyle{IEEEtran}
\end{document}